\documentclass[10pt]{article}

\usepackage{hyperref}

\usepackage{amsmath, amsthm, amssymb, mathabx}
\usepackage{bussproofs}
\usepackage{enumerate}
\usepackage{graphicx}
\usepackage{mathrsfs}

\DeclareMathOperator{\var}{var}

\newcommand{\imp}{\rightarrow}
\newcommand{\Ra}{\Rightarrow}

\newcommand{\qo}{\rightbarharpoon_{\Omega}}  
\newcommand{\qohyp}{\preccurlyeq_{\text{hyp}}}

\newcommand{\maj}{\preccurlyeq_{maj}}
\renewcommand{\min}{\preccurlyeq_{min}}

\newcommand{\pow}[1]{P_f({\mathbb{N}^{#1}})}
\newcommand{\powf}[1]{P_f(#1)}
\newcommand{\nn}{\mathbb{N}}
\newcommand{\minimum}{\text{min}}
\newcommand{\mply}{.}

\newcommand{\maxelesize}[1]{\langle#1\rangle}
\renewcommand{\max}{\text{max}}

\newtheorem{theorem}{Theorem}

\newtheorem{example}[theorem]{Example}
\newtheorem{lemma}[theorem]{Lemma}

\newtheorem{definition}[theorem]{Definition}

\newcommand{\WI}{\text{WI}}

\newcommand{\fus}{\cdot}
\newcommand{\FLe}{\text{FL}_{\text{e}}}

\newcommand{\HFLe}{\text{HFL}_{\text{e}}}
\newcommand{\FLec}{\text{FL}_{\text{ec}}}
\newcommand{\FLelw}{\text{FL}_{\text{elw}}}
\newcommand{\FLecm}{\text{FL}_{\text{ecm}}}
\newcommand{\FLew}{\text{FL}_{\text{ew}}}
\newcommand{\FLecw}{\text{FL}_{\text{ecw}}}
\newcommand{\HFLec}{\text{HFL}_{\text{ec}}}
\newcommand{\HFLew}{\text{HFL}_{\text{ew}}}
\newcommand{\HFLelw}{\text{HFL}_{\text{elw}}}
\newcommand{\h}{H}

\newcommand{\reduct}[1]{r(#1)}
\newcommand{\DC}{\Ra}
\newcommand{\VL}{\,|\,}
\newcommand{\N}{\mathcal{N}}
\renewcommand{\P}{\mathcal{P}}
\newcommand{\landone}[1]{(#1)_{\land 1}}
\newcommand{\C}{\mathcal{C}}
\newcommand{\minus}[1]{#1^{-}}

\newcommand{\MTL}{\text{MTL}}
\newcommand{\WNM}{\text{WNM}}

\newcommand{\hypsize}[1]{\langle#1\rangle}

\newcommand{\thin}{thin}
\newcommand{\slim}{slim}

\author{
A. R. Balasubramanian\\
Technische Universit\"{a}t Munchen\\
\texttt{bala.ayikudi@tum.de}
\and
Timo Lang\\
Technische Universit\"{a}t Wien\\
\texttt{timo@logic.at}
\and
Revantha Ramanayake\\
University of Groningen\\
\texttt{d.r.s.ramanayake@rug.nl}
}

\date{}

\begin{document}

\title{Decidability and Complexity in Weakening and Contraction Hypersequent Substructural Logics\footnote{This is the authors' version of the work. It is posted here for your personal use. Not for redistribution.}}

\maketitle

\begin{abstract}
We establish decidability for the infinitely many axiomatic extensions of the commutative Full Lambek logic with weakening $\FLew$ (i.e. $\text{IMALLW}$) that have a cut-free hypersequent proof calculus (specifically: every analytic structural rule extension). 
Decidability for the corresponding extensions of its contraction counterpart $\FLec$ was established recently but their computational complexity was left unanswered. In the second part of this paper, we introduce just enough on length functions for well-quasi-orderings and the fast-growing complexity classes to obtain complexity upper bounds for both the weakening and contraction extensions.
A specific instance of this result yields the first complexity bound for the prominent fuzzy logic $\text{MTL}$ (monoidal t-norm based logic) providing an answer to a long-standing open problem.
\end{abstract}

\section{Introduction}

Logical systems (or simply, logics) model the reasoning that applies within various concepts. Two familiar examples are classical logic (modelling truth) and intuitionistic logic (modelling constructive proof).
A \textit{substructural logic} lacks some of the properties (`structural rules') of these logics. The most notable structural rules are weakening (insert an arbitrary hypothesis), contraction (delete a copy of a hypothesis that occurs multiple times), commutativity/exchange (swap the position of hypotheses), and associativity. 

Let us demonstrate the motivation for omitting structural rules via some examples. Interpret ``$\$,\$\Ra \text{ticket}$" as ``with two dollars I can buy the ticket". Applying the contraction rule to the latter would yield ``$\$\Ra\text{ticket}$" i.e.\ ``with one dollar I can buy the ticket".
Typically we would not want to permit this inference. Hence we would want to reject the contraction rule when we have such a resource conscious interpretation in mind. Meanwhile, if we interpret ``$P\Ra C$" as ``$C$ truly depends on $P$'' then from this we would not want to infer ``$P,Q\Ra C$"  (i.e.\ ``$C$ truly depends on $P$ together with $Q$''). Hence we would want to reject the weakening rule when we have such an interpretation in mind.

Put simply, the choice of which structural rules to retain and the assertion of further properties (axioms) lead to an infinite number of different substructural logics which are able to model a wide range of notions. Indeed, concerning the axioms, a truly astonishing number have been classified and studied in different contexts by various research communities.

Owing to their versatility, substructural logics, along with modal logics, provide a powerful tool for modelling and reasoning in computer science. Linear logic and its many variants (computational and resource consciousness), extensions of the Lambek calculus (syntax and syntactic types of natural language, context-free grammars, computational linguistics), fuzzy logics (formal reasoning about vagueness, fuzzy systems modelling), and bunched implication logics (software program verification, static analysis of run-time memory allocation) are just a few examples.

Throughout this work, we identify a logic with the set of formulas representing its properties i.e. its theorems. 
A logic is \textit{decidable} if there is an algorithm that can determine if an input formula is a theorem of the logic. For a decidable logic, the natural question is its \textit{computational complexity}: how much time and space is required to run the algorithm as a function of the input size? 
Decidability and complexity are such fundamental properties that these questions must surely have been raised for every single logic that has been studied.

A prominent method of studying a logic is via proof theory, by using a proof calculus. A proof calculus is a mathematical object that generates (finite) proofs of exactly those formulas that are theorems of the logic. 
 The typical approach to establish decidability is to attempt to build a proof backwards from the input and abort at some finite point if no proof exists (`backward proof search'). This approach relies on a proof calculus whose proofs have nice properties: most crucially, unknown formulas should not appear when building the proof backwards.
Typically one asks that the proof calculus has the famed \textit{subformula property}: every formula that occurs in a proof is a subformula of the formula being proved.
 
The sequent calculus is the type of proof calculus introduced by Gentzen in 1935. It consists of a finite set of proof rules defined on \textit{sequents}. A sequent generalises a formula and takes the form $X\Ra\Pi$ ($X\cup \Pi$ is a finite multiset of formulas and $\Pi$ contains at most one formula).
The sequent calculus for commutative Full Lambek logic (logic and proof calculus are both denoted~$\FLe$), equivalently intuitionistic multiplicative additive linear logic, is the starting point for our discussion. The proof rules for contraction and weakening are as follows. The comma is read as multiset sum following convention.
\begin{center}
\begin{tabular}{cc}
\AxiomC{$X,A,A\Ra\Pi$}
\RightLabel{(c)ontraction}
\UnaryInfC{$X,A\Ra\Pi$}
\DisplayProof
&
\AxiomC{$X\Ra\Pi$}
\RightLabel{(w)eakening}
\UnaryInfC{$X,Y\Ra\Pi,\Pi'$}
\DisplayProof
\end{tabular}
\end{center}
As simple as it looks, the contraction rule is famously hard to control\footnote{Girard~\cite{Gir95} calls it the ``fingernail of infinity in the propositional calculus".} in backward proof search because the premise (the sequent above the horizontal line) is larger than the conclusion. In contrast, weakening is much easier to handle. Indeed, compare Kripke's~\cite{Kri59} famous proof of decidability for $\FLe+(c)=\FLec$ (the logic was shown to be non-primitive recursive~\cite{Urq99}) with the straightforward decidability argument for $\FLe+(w)=\FLew$ (PSPACE-complete~\cite{HorTer11}). However, when it comes to \textit{extensions of} $\FLec$ and $\FLew$, it is the former that are easier to control. This is because the contraction rule can be used to prune the backward proof search tree by dismissing the infinitely many sequents larger than a given sequent (structural proof theorists would say: using height-preserving admissibility of contraction). An analogous use of the weakening rule dismisses the sequents smaller than a given sequent, and there are only finitely many of these. 

The above is illustrated in the following backward proof search tree from $A^2,B^3 \Ra F$ (i.e. $A,A,B,B,B\Ra F$), which makes use of a rule which we shall call (scom).%
\footnote{Read bottom-up, (scom) doubles the multiplicity of each formula in the conclusion antecedent and distributes the formulas between two premises. (scom) can be seen as a sequent version of the \textit{communication rule} (com). The latter is prominent in the proof theory of G\"odel logic (see, e.g., \cite{MetOliGab09}). 
}
\begin{center}
\begin{small}
\AxiomC{$A^3,B^1\Ra F$\hspace{-0.5cm}}
\AxiomC{$\ddag$ $A^4,B^2\Ra F$}
\AxiomC{$A^2,B^4\Ra F$}
\RightLabel{scom}
\BinaryInfC{$\dag$ $A^3,B^3\Ra F$}
\RightLabel{scom}
\BinaryInfC{$\ddag$ $A^3,B^2\Ra F$}		
\AxiomC{\hspace{-2cm}$A,B^4\Ra F$}
\RightLabel{scom}
\BinaryInfC{$\dag$ $A^2,B^3 \Ra F$}
\DisplayProof
\end{small}
\end{center}
In the presence of contraction, it can be argued that it is possible to obtain a finite backward proof search tree by dismissing any sequent that is contractible (i.e. repeated applications of the contraction rule) to some sequent below it.\footnote{An \textit{everywhere minimal proof} is a proof whose every subproof has minimal height. Suppose a proof~$d$ contains $A^3,B^3\Ra F$ somewhere above $A^2,B^3 \Ra F$. The former is contractible to the latter so any proof~$d_{1}$ of $A^3,B^3\Ra F$ yields a proof~$d_{1}'$ \textit{of the same height} of $A^2,B^3 \Ra F$ by height-preserving admissibility of contraction. It follows that $d$ is not everywhere minimal, since its subproof of $A^2,B^3 \Ra F$ has greater height than~$d_{1}'$. Since it can be shown that every provable sequent has an everywhere minimal proof, it is safe to disregard~$d$, and hence also a backward proof search tree containing a sequent contractible to some sequent below it.}
For example, we can dismiss $A^3,B^3\Ra F$ because it is contractible to $A^2,B^3\Ra F$.
This is the basis of the decidability argument for $\FLec$. 
On the other hand, no sequent in the tree can be weakened to obtain any sequent below it.
So if contraction is replaced by weakening there is no apparent justification to prohibit the backward proof search from extending the tree indefinitely by obtaining $A^{2+n},B^3\Ra F$ ($\dag$) and $A^{3+n},B^2\Ra F$ ($\ddag$) for every~$n$.
This is not an issue for $\FLew$ since its proof calculus does not contain a rule like (scom). 
It \textit{is} an issue for many of its extensions; $\FLew+(scom)$ is just one example.

There is in fact a significant challenge to be overcome \textit{before} pondering how to handle the backward proof search tree: finding a proof calculus \textit{with} the subformula property for the logic of interest. This is a major preoccupation of structural proof theory. Here we are in luck because Ciabattoni \textit{et al.}~\cite{CiaGalTer08} constructed \textit{hypersequent calculi} with the subformula property for an infinite set of substructural logics, and many logics of interest are in this set. In this type of proof calculus, proof rules are defined on hypersequents (multisets of sequents). We refer to such logics as \textit{hypersequent substructural logics}.

As a consequence of their additional structure, decidability arguments are ``further complicated"~\cite{MetOliGab09} when hypersequent calculi are employed (this will be evidenced later on in the formal complexity analysis). Nevertheless, Ramanayake~\cite{Ram20LICS} recently showed the decidability of every hypersequent substructural logic extending $\FLec$ by extending the argument sketched above to hypersequent calculi. However that proof was not constructive enough to obtain complexity bounds.

\textbf{In this work} we establish the decidability for every hypersequent substructural logic extending $\FLew$.
Since backward proof search seemed fruitless, we develop a forward strategy--- unlike the former there is no standard methodology---where only a limited amount of weakening is permitted above a premise of a rule; anything more must be applied after its conclusion.
We define a well-quasi-ordering on hypersequents and use this to show that only finitely many forward steps are required to determine if the input has a proof.  In the second part of the paper we obtain complexity upper bound for these logics. Finally, we refine the argument in~\cite{Ram20LICS} to extract complexity upper bounds also for the hypersequent substructural logics extending $\FLec$.

Related work: in addition to the already mentioned~\cite{Ram20LICS}, the following works present decidability and complexity results that apply to multiple classes of extensions of~$\FLe$. Each makes significant concessions: Galatos and Jipsen~\cite{GalJip13} and St. John~\cite{StJ19} 
consider very specific axiom forms and the extensions are restricted to sequent calculi; Ciabattoni \textit{et al.}~\cite{CiaLanRam19} consider hypersequent substructural logics but only for extensions of $\FLecm$ ($m$ is the mingle axiom and it is a specific instance of weakening).

We conclude by expanding on three further aspects:

\subsubsection*{Complexity} Urquhart~\cite{Urq99} gave tight Ackermannian bounds $\mathbf{F}_{\omega}$ for $\FLec$. To understand how such an upper bound arises, consider the backward proof search tree above. We noted that $A^3,B^3\Ra F$ can be dismissed from the tree because it is contractible to a sequent $A^2,B^3 \Ra F$ below it. What sequents can we \textit{not} dismiss with respect to $A^2,B^3 \Ra F$? Certainly the finitely many smaller sequents $A^1,B^3 \Ra F$, $A^2,B^2 \Ra F$, $\ldots$. Also sequents that are incomparable to it, such as $A^1,B^{100} \Ra F$. Although there are infinitely many such sequents, there is still hope! After all, only finitely many can be obtained in a single step since there is a fixed polynomial bounding the size of a premise from every proof rule in terms of the size of its conclusion. It turns out that this is enough to bound the height of the tree. Equivalently, there is an upper bound on the maximum length of non-increasing non-constant (`bad') sequences under the componentwise ordering (\textit{length function theorem}).
(For the lower bounds, Urquhart encodes a problem with known complexity into the logic.)

We obtain a hyper-Ackermannian $\mathbf{F}_{\omega^{\omega}}$-upper bound for hypersequent substructural extensions of~$\FLec$ and~$\FLew$, by exploiting Balasubramanian's~\cite{Bal20} recent upper bounds for bad sequences on the majoring and minoring orderings. Of course, many specific logics among these extensions---like intuitionistic logic~$\FLecw$---are known to have much more modest bounds. 

Since the above complexity classes and methods are less widely known, we provide a gentle introduction in Section~\ref{sec-upper-bound}.

\subsubsection*{Well-quasi-orderings (wqo)} Wqos play an important role in computer science, for example to show the termination of algorithms and term-rewriting systems. This work is far from the first to use them within logic.
Nevertheless, this work illustrates a methodology for proof theory: use the structural proof theory to present the logic in a suitable form---e.g. height-preserving admissibility lemmas, rule permutations, absorbing one rule into another---in order that a wqo can be identified on the basic units (sequents, hypersequents, \ldots) of the proof calculus. Now utilise the wqo to obtain decidability and complexity. 
To emphasise this methodology, we take care to separate the contributions from the proof calculus and from the wqo length function bound in the complexity calculation in Section~\ref{sec-upper-bound}.

\subsubsection*{Fuzzy logics} Many of the logics covered by this result are of independent interest, and the others are candidates for future applications. The latter is pertinent due to the widespread applicability of the Lambek calculus and its extensions. For many of the logics, these are the first decidability and complexity results. As we shall see, these observations are especially relevant for mathematical fuzzy logics, which provide a formal basis for some of Zadeh's fuzzy logics~\cite{Zah65}, and fuzzy systems modelling (Yager and collaborators e.g.~\cite{YagRyb96}).

The monoidal t-norm based logic $\MTL$  was introduced by Esteva and Godo in 2001~\cite{EstGod01}. It is axiomatised by extending $\FLew$ with prelinearity $(p\imp q)\lor (q\imp p)$. 

To explain the prominence of this logic we need to delve into the design of a mathematical fuzzy logic. H\'ajek's~\cite{Haj98} insight was that the definition of the fuzzy conjunction ($\cdot$ in the notation of this paper) is the crucial building block for developing a formal framework for fuzzy logics. He proposed to use a t(riangular)-norm---also used in Zadeh's fuzzy sets---on the $[0,1]$-unit interval for the fuzzy conjunction because it supported the desired philosophical desiderata. This interval is interpreted as the truth degrees with $1$ read as classical truth and $0$ as classical false. 
Left-continuity is also demanded of the t-norm since this is necessary and sufficient for the residuum to be defined: this becomes the fuzzy implication.
Meanwhile, from the syntactic perspective, weakening is the crucial ingredient for axiomatising t-norms.

$\MTL$ is prominent because of the importance of the above features to fuzzy logics. Specifically, it describes the common behaviours of \textit{all} fuzzy logics based on left-continuous t-norms. Indeed, Horc{\'{\i}}k \textit{et al.}~\cite{HorNogPet07} observe  ``[$\MTL$ is] the weakest fuzzy logic\footnote{It should be noted that there are also other candidates for this title: uninorm logic~\cite{MetMon07} or the weakly implicative semilinear logics~\cite{CinNog10}.} and the research on fuzzy logic systems becomes research on extensions of $\MTL$".

This work presents the first syntactic proof of decidability for $\MTL$. The existing proof of decidability---due to Ono following the argument in~\cite{BlovanAlt02}; see~\cite{CiaMetMon10} for a proof---relied on the algebraic semantics and there was ``no known complexity bound"~\cite{MetOliGab09} (see also the Handbook chapter~\cite{Han11}). Hanikov\'{a}~\cite{Han17} observes {``we would like to know this for $\MTL$, whose complexity is a long-standing open problem within propositional fuzzy logics"}. We answer this open problem by presenting the first upper bound for $\MTL$. 

The results also apply to several genuinely fuzzy (`standard complete') axiomatic extensions of~$\MTL$ studied in the past two decades, for example:
\begin{itemize}
\item $n$-contractive extensions $\text{C}_n\MTL = \MTL + p^{n-1} \imp p^n$ ($n\geq 2$); 
decidable~\cite{HorNogPet07}
\item weak nilpotent minimum logic $\WNM = \MTL + \lnot(p\fus q)\lor ((p\land q)\imp p\fus q)$; 
decidable~\cite{NogEstGis08}, co-NP-completeness for every finitely axiomatisable extensions of~$\WNM$~\cite{EstGodNog10}
\item strict monoidal t-norm based logic $\MTL+(p\land\lnot p)\imp 0$; 
decidable~\cite{CiaMetMon10}
\item $\MTL+(wmn)^n$ where $(wmn)^n:= \lnot(p\fus q)^{n}\lor ((p\land q)^{n-1}\imp (p\fus q)^{n})$ and $n\geq 2$
\end{itemize}
Even in the cases where decidability was already known, the above works use distinctive algebraic semantic arguments. 
In contrast, the argument in this work applies in one shot to all of these classes. Decidability for $\MTL+(wmn)^n$ is new. So are the complexity bounds for $\text{C}_n\MTL$ and $\MTL+(wmn)^n$.

\section{Preliminaries}

Let~$|\Omega|$ denote the cardinality of the set~$\Omega$.
A \textit{multiset} of a set~$A$ is a map $M:A\mapsto\mathbb{N}$;
$M(a)$ is called the \textit{multiplicity} of the element~$a\in A$.
The multiset is finite if only finitely many elements of~$A$ have positive multiplicity.
The cardinality of a finite multiset is the sum of the multiplicities of its elements.
The sum of multisets~$M_{1}$ and~$M_{2}$ (of some set~$A$) is the multiset given by the map $a\mapsto M_{1}(a)+M_{2}(a)$ for $a\in A$.

Let~$\mathsf{Var}$ be a countably infinite set of \textit{propositional variables}.
Logical formulas are defined by the following grammar.  
\[
\mathsf{Form}:=p\in \mathsf{Var}|\top|\bot|1|0|(F\land F)|(F\lor F)|(F\fus F)|(F\imp F)
\]
The connective~$\fus$ is called \textit{fusion} or \textit{times}. It is also called multiplicative conjunction to contrast it with the additive conjunction~${\land}$. We often omit leading parentheses to reduce clutter, e.g. writing $(\top\land (\bot\lor\bot))$ as $\top\land (\bot\lor\bot)$.

A \textit{logic}~$L$ is a set of formulas (`theorems') from~$\mathsf{Form}$ that is closed under the uniform substitution of formulas for propositional variables, and closed under \textit{modus ponens}: $A\in L$ and $A\imp B\in L$ implies $B\in L$.
The \textit{axiomatic extension}~$L+\mathcal{F}$ of the logic~$L$ by a finite set~$\mathcal{F}$ of formulas is defined in the usual way as 
the smallest logic containing~$L\cup \mathcal{F}$.


\subsection{Basic definitions from structural proof theory}

A \textit{sequent} is a tuple written~$X\Ra\Pi$ where~$X$ (the \textit{antecedent}) is a finite multiset of formulas and~$\Pi$ (the \textit{succedent}) is a multiset that contains at most a single formula.

A \textit{hypersequent} is a finite multiset (possibly empty) of sequents.
It is often explicitly written as a list of $|$-separated sequents as follows:
\begin{equation}\label{hypersequent}
X_{1} \Ra  \Pi_{1} | \ldots | X_{n} \Ra  \Pi_{n}
\end{equation}
Each sequent~$X_{i}\Ra \Pi_{i}$ is said to be a \textit{component} of the hypersequent.
In practice, ``sequent" and ``component" are often used interchangeably.

Define~$\hypsize{h}$ to be the number of symbols in the standard written representation of the hypersequent~$h$. For example, $\hypsize{p,p\Ra |\Ra q\land p}=9$. Nevertheless the precise details of the counting do not matter for this paper.


A \textit{hypersequent calculus} is a type of formal proof calculus (introduced independently in~\cite{Min68,Pot83,Avr87}) that is used to generate proofs (`derivations') of hypersequents. It is a generalisation of the \textit{sequent calculus} introduced by Gentzen~\cite{Gen69}. Formally, a hypersequent calculus is a finite set of \textit{hypersequent rule schemas} of the following form
where $h_{0}$ is the \textit{conclusion} and $h_1,\ldots,h_n$ ($n\geq 0$) are the \textit{premise(s)}.
\begin{center}
\AxiomC{$h_1$}
\AxiomC{$\ldots$}
\AxiomC{$h_n$}
\TrinaryInfC{$h_{0}$}
\DisplayProof
\end{center}
Each~$h_{i}$ is called a \textit{schematic-hypersequent} and has the following form for $k\geq 0$:
\[
H | \mathcal{L}_{1}\Ra \mathcal{M}_{1} | \cdots | \mathcal{L}_{k}\Ra \mathcal{M}_{k}
\]
In the above: $H$~is the \textit{hypersequent-variable}, each~$\mathcal{L}_{i}$ is a list comprising of \textit{multiset-variables}, \textit{formula-variables}, and \textit{schematic-formulas} (itself built from the logical connectives and constants using formula-variables), and each~$\mathcal{M}_{i}$ is either empty, a \textit{succedent-variable}, or a \textit{schematic-formula}. 

A rule schema with no premises is an \textit{initial rule schema}. 
A rule schema comprising of just a hypersequent-variable, multiset-variables, and succedent-variables is called a \textit{structural rule schema}.
Define~$\hypsize{r}$ to be the number of symbols in the standard written representation of the rule schema.
\begin{example}\label{eg-landR}
The rule schema~(${\land}$R) below is not a structural rule since it contains formulas-variables ($A,B$) and indeed a schematic-formula ($A\land B$) as well.
\[
\AxiomC{$\h |  X\DC A$}
\AxiomC{$\h |  X\DC B$}
\RightLabel{(${\land}$R)}
\BinaryInfC{$\h |  X\DC A\land B$}
\DisplayProof
\]
The following is an example of a structural rule schema. 
\[
\AxiomC{$\h |  X_{1},Y_{1}\Ra\Pi_{1}$}
\AxiomC{$\h |  X_{2},Y_{2}\Ra\Pi_{2}$}
\RightLabel{(com)}
\BinaryInfC{$\h |  X_{1},Y_{2}\Ra\Pi_{1} | X_{2},Y_{1}\Ra\Pi_{2}$}
\DisplayProof
\]
%
\end{example} 
%
Every schematic-variable is intended for instantiation by a certain type of object.
\begin{center}
\begin{tabular}{l | l}
schematic-variable (notation)				& type of instantiation \\ \hline
hypersequent-variable ($H$)				& any hypersequent (also empty)\\
multiset-variable ($X,Y,Z$)				& any multiset of formulas			\\
succedent-variable ($\Pi$)					& any formula or empty				\\
formula-variable ($A,B$)					& any formula
\end{tabular}
\end{center}
Let~$r$ be a rule schema.
An \textit{instantiation}~$I$ of~$r$ is a map from each schematic-variable occurring in~$r$---let~$\var(r)$ denote the set of such schematic-variables---to an object of the corresponding type.
A \textit{rule instance}~$I(r)$ is obtained by instantiating each~$\alpha\in\var(r)$ with~$I(\alpha)$; the instantiation of a schematic-formula is determined by the instantiation of its constituent formula-variables. E.g.~$I(A\heartsuit B):=I(A)\heartsuit I(B)$.

This definition of rule schema and rule instance is the standard one from structural proof theory. Usually this distinction is not made explicit, since it can be discerned from its context. We have given this formal development because it will be helpful for formulating precise arguments later on. 
\begin{example}
Consider the rule schema~(${\land}$R) from Eg.~\ref{eg-landR}. Then $\var(\text{${\land}$R})$ consists of the hypersequent-variable~$H$, the multiset-variable~$X$, and the formula-variables~$A$ and~$B$. Consider the following instantiations of~(${\land}$R).
\begin{align*}
&I_{1}(H)=\emptyset		&& I_2(H)=\emptyset		&& I_3(H)= \quad \Ra s | q\imp p\Ra p	\\
&I_{1}(X)=\emptyset		&& I_2(X)=\{r,r\}		&& I_3(X)= \{r\}					\\
&I_{1}(A)=p				&& I_2(A)=p\land q		&& I_3(A)= p\land q				\\
&I_{1}(B)=q				&& I_2(B)=q			&& I_3(B)= q
\end{align*}
Rule instances~$I_1(\text{${\land}$R})$, $I_2(\text{${\land}$R})$, and $I_3(\text{${\land}$R})$ appear below.
\begin{center}
\begin{small}
\begin{tabular}{c@{\hspace{0.3cm}}c}
\AxiomC{$\Ra p$}
\AxiomC{$\Ra q$}
\RightLabel{$I_1(\text{${\land}$R})$}
\BinaryInfC{$\Ra p\land q$}
\DisplayProof
&
\AxiomC{$r,r\Ra p\land q$}
\AxiomC{$r,r\Ra q$}
\RightLabel{$I_2(\text{${\land}$R})$}
\BinaryInfC{$r,r \Ra (p\land q)\land q$}
\DisplayProof
\\[1.5em]
\multicolumn{2}{c}{
\AxiomC{$\Ra s | q\imp p\Ra p| r\Ra p\land q$}
\AxiomC{$\Ra s | q\imp p\Ra p| r\Ra q$}
\RightLabel{$I_3(\text{${\land}$R})$}
\BinaryInfC{$\Ra s | q\imp p\Ra p| r \Ra (p\land q)\land q$}
\DisplayProof
}
\end{tabular}
\end{small}
\end{center}
\end{example}

A \textit{derivation} of the hypersequent~$h$ in the hypersequent calculus~$\mathcal{H}$ is defined in the usual way as a finite tree of hypersequents such that its root is~$h$, its leaves are instances of initial rule schemas, and each interior node and its children are the conclusion and premises of an instance of some rule schema in~$\mathcal{H}$.

A derivation of the formula $B$ is a derivation of the hypersequent $\Ra B$.

The \textit{height} of a derivation is the number of nodes on its longest branch.

The \textbf{hypersequent calculus}~$\HFLe$ is given in Fig.~\ref{figure-HFLec}.
An example of derivation in~$\HFLe$ is given in Fig.~\ref{fig-sample-der}.
The cut-rule below is \textit{not} a rule schema in~$\HFLe$ but it is well-known to be admissible (i.e. if the premises of an instance of the cut-rule are derivable in~$\HFLe$ then so is the conclusion).
\begin{center}
\AxiomC{$H | X,A\Ra\Pi$}
\AxiomC{$H | Y\Ra A$}
\RightLabel{(cut)}
\BinaryInfC{$H | X,Y\Ra \Pi$}
\DisplayProof
\end{center}


The \textit{extension of the hypersequent calculus}~$\mathcal{H}$ by the finite set~$R$ of rule schemas is the hypersequent calculus~$\mathcal{H}\cup R$ (following standard convention, we write~$\mathcal{H}+R$). 

Here are the rule schemas of contraction~(c) and left weakening~(lw). In the commutative setting of this paper each is interchangeable with the structural rule schema obtained by replacing the formula-variable~$A$ with a multiset-variable~$Y$.
\begin{center}
\begin{tabular}{c@{\hspace{1cm}}c}
\AxiomC{$\h |  X,A,A\DC \Pi$}
\RightLabel{(c)}
\UnaryInfC{$\h |  X,A\DC \Pi$}
\DisplayProof
&
\AxiomC{$\h |  X\DC \Pi$}
\RightLabel{(lw)}
\UnaryInfC{$\h |  X,A\DC \Pi$}
\DisplayProof
\end{tabular}
\end{center}
The rule schemas (lw) and right weakening (rw) below left can be combined as the weakening rule below right (there are other variants that are also equivalent).
This paper studies extensions of~$\HFLe+(lw)$. The extensions $\HFLe+(lw)+(rw)$/$\HFLe+(w)$ can then be viewed as a particular case of the general result.
\begin{center}
\begin{tabular}{c@{\hspace{0.5cm}}c}
\AxiomC{$\h |  X\DC $}
\RightLabel{(rw)}
\UnaryInfC{$\h |  X\DC A$}
\DisplayProof
&
\AxiomC{$\h |  X\DC\Pi$}
\RightLabel{(w)}
\UnaryInfC{$\h |  X,Y\DC\Pi,\Pi'$}
\DisplayProof
\end{tabular}
\end{center}
We denote $\HFLe+(c)$ by~$\HFLec$, $\HFLe+(lw)$ by~$\HFLelw$, and $\HFLe+(lw)+(rw)$ by~$\HFLew$.

We say that~$\mathcal{H}$ is a \textit{hypersequent calculus for the logic}~$L$ if
\[
\text{For every formula~$B$: $B\in L$ iff $\mathcal{H}$ derives $\Ra B$}
\]
 
The \emph{decision problem} for a hypersequent calculus~$\mathcal{H}$ asks
\begin{center}
Is a given hypersequent~$h$ derivable in~$\mathcal{H}$?
\end{center}

\begin{figure*}
\hspace{-1cm}
\begin{scriptsize}
\begin{tabular}{ccc}
\multicolumn{3}{c}{
\begin{tabular}{c@{\hspace{3em}}c@{\hspace{3em}}c@{\hspace{3em}}c@{\hspace{3em}}c@{\hspace{3em}}c}
\AxiomC{}
\UnaryInfC{$\h |  A\DC A$}
\DisplayProof
&
\AxiomC{}
\UnaryInfC{$\h |  \bot, X\Ra\Pi$}
\DisplayProof
&
\AxiomC{}
\UnaryInfC{$\h |  X\Ra \top$}
\DisplayProof
&
\AxiomC{}
\UnaryInfC{$\h |  0\Ra$}
\DisplayProof
&
\AxiomC{}
\UnaryInfC{$\h |  \Ra 1$}
\DisplayProof
&
\AxiomC{$\h |  X\Ra \Pi$}
\UnaryInfC{$\h |  1,X\Ra \Pi$}
\DisplayProof
\end{tabular}
}
\\[1.5em]
\AxiomC{$\h |  X\Ra$}
\UnaryInfC{$\h |  X\Ra 0$}
\DisplayProof
&
\AxiomC{$\h |  X\Ra\Pi | X\Ra\Pi$}
\RightLabel{(EC)}
\UnaryInfC{$\h |  X\Ra\Pi$}
\DisplayProof
&
\AxiomC{$\h$}
\RightLabel{(EW)}
\UnaryInfC{$\h |  X\Ra\Pi$}
\DisplayProof
\\[1.5em]
\AxiomC{$\h |  X,A,B\DC \Pi$}
\RightLabel{($\fus$L)}
\UnaryInfC{$\h |  X,A\fus B\DC \Pi$}
\DisplayProof
&
\AxiomC{$\h |  X\DC A$}
\AxiomC{$\h |  Y\DC B$}
\RightLabel{($\fus$R)}
\BinaryInfC{$\h |  X,Y\DC A\fus B$}
\DisplayProof
&
\AxiomC{$\h |  X,A\DC \Pi$}
\AxiomC{$\h |  X,B\DC \Pi$}
\RightLabel{($\lor$L)}
\BinaryInfC{$\h |  X,A\lor B\DC \Pi$}
\DisplayProof
\\[1.5em]
\AxiomC{$\h |  X\DC A_{i}$}
\RightLabel{($\lor$R)}
\UnaryInfC{$\h |  X\DC A_{1}\lor A_{2}$}
\DisplayProof
&
\AxiomC{$\h |  X,A_{i}\DC \Pi$}
\RightLabel{(${\land}$L)}
\UnaryInfC{$\h |  X,A_{1}\land A_{2}\DC \Pi$}
\DisplayProof
&
\AxiomC{$\h |  X\DC A$}
\AxiomC{$\h |  X\DC B$}
\RightLabel{(${\land}$R)}
\BinaryInfC{$\h |  X\DC A\land B$}
\DisplayProof
\\[1.5em]
\multicolumn{3}{c}{
\begin{tabular}{cc}
\AxiomC{$\h |  X\DC A$}
\AxiomC{$\h |  Y,B\DC \Pi$}
\RightLabel{(${\imp}$L)}
\BinaryInfC{$\h |  X,Y,A\imp B\DC \Pi$}
\DisplayProof
&
\AxiomC{$\h |  X,A\DC B$}
\RightLabel{(${\imp}$R)}
\UnaryInfC{$\h |  X\DC A\imp B$}
\DisplayProof
\end{tabular}
}
\end{tabular}
\end{scriptsize}
\caption{The hypersequent calculus~$\HFLe$ for~$\FLe$}
\label{figure-HFLec}
\end{figure*}

\begin{figure*}
\begin{center}
\AxiomC{}
\UnaryInfC{$\Ra p | p\Ra p$}
\AxiomC{}
\UnaryInfC{$\Ra p | q\Ra q$}
\RightLabel{($\fus$R)}
\BinaryInfC{$\Ra p | p,q\Ra p\fus q$}
\RightLabel{$\lor$R}
\UnaryInfC{$\Ra p | p,q\Ra (p\fus q)\lor(p\fus r)$}
\AxiomC{}
\UnaryInfC{$\Ra p | p\Ra p$}
\AxiomC{}
\UnaryInfC{$\Ra p | r\Ra r$}
\RightLabel{($\fus$R)}
\BinaryInfC{$\Ra p | p,r\Ra p\fus r$}
\RightLabel{($\lor$R)}
\UnaryInfC{$\Ra p | p,r\Ra (p\fus q)\lor(p\fus r)$}
\RightLabel{$\lor$L}
\BinaryInfC{$\Ra p | p,q\lor r\Ra (p\fus q)\lor(p\fus r)$}
\RightLabel{($\fus$L)}
\UnaryInfC{$\Ra p | p\fus (q\lor r)\Ra (p\fus q)\lor(p\fus r)$}
\DisplayProof
\end{center}
\caption{An example of a derivation in $\HFLe$. Derivation of $\Ra p | p\fus (q\lor r)\Ra(p\fus q)\lor(p\fus r)$.}
\label{fig-sample-der}
\end{figure*}

\subsection{Hypersequent calculi for substructural logics} Ciabattoni \textit{et al.}~\cite{CiaGalTer08,CiaGalTer17}
present hypersequent calculi with the subformula property for many axiomatic extensions of~$\FLe$, by extending~$\HFLe$
with \textit{analytic structural rule schemas}.
\begin{definition}\label{def-linear-subvariable}
A rule schema has a \emph{linear conclusion} if each schematic-variable in the conclusion occurs exactly once there;
it has the \emph{subvariable property} if every schematic-variable in the premise occurs in the conclusion.
\end{definition}
Each analytic structural rule schema has a linear conclusion and the subvariable property. This is all that matters for this paper. Nevertheless, here is the formal definition.
\begin{definition}
An \emph{analytic structural rule schema} has the form below and it is built from the hypersequent-variable~$H$, pairwise distinct succedent-variables~$\{\Pi_{i}|i\in I\}$, and pairwise distinct multiset-variables $\{Y_{i}| i\in I\}\cup \{X_{is} | i\in I, 1\leq s\leq s_i\}\cup \{Z_{jt} | j\in J, 1\leq t\leq t_j\}$. Each~$\mathcal{L}_{ik}$ ($i\in I, k\in K_i$) and~$\mathcal{M}_{l}$ ($l\in L$) is a list from the multiset-variables.
\begin{equation}\label{analytic-rule}
\text{
\AxiomC{$\{ H \VL Y_{i}, \mathcal{L}_{ik}\Ra \Pi_{i} \}_{i\in I, k\in K_{i}}$}
\AxiomC{$\{ H \VL \mathcal{M}_{l}\Ra \}_{l\in L}$}
\BinaryInfC{$H \VL Y_{i},X_{i1},\ldots,X_{i s_i}\Ra\Pi_{i} (i\in I) \VL Z_{j1},\ldots,Z_{j t_j}\Ra (j\in J)$}
\DisplayProof
}
\end{equation}
\end{definition}
Evidently every rule schema in \textit{every analytic structural rule extension of~$\HFLec$ and~$\HFLelw$} has a linear conclusion and the subvariable property. The \textit{subformula property} follows for each calculus: every formula occurring in a derivation of a hypersequent~$h$ is a subformula of some formula in~$h$.

The \textit{substructural hierarchy}~\cite{CiaGalTer08} identifies the logics that have an analytic structural rule extension of~$\HFLe$. First define $\P_{0}=\N_{0}=\mathsf{Var}$ (propositional variables). Now define
\begin{align*}
&\P_{n+1} := 1 \VL \bot \VL \N_{n} \VL \P_{n+1}\lor\P_{n+1}\VL \P_{n+1}\fus \P_{n+1} 		\\
&\N_{n+1} := 0 \VL \top \VL \P_{n} \VL \N_{n+1}\land\N_{n+1}\VL \P_{n+1}\imp \N_{n+1}	\\
&\text{Also define }\P_{3}' := 1\VL \bot \VL \N_{2}\land 1 \VL \P_{3}'\lor \P_{3}'  \VL \P_{3}'\fus \P_{3}'
\end{align*}
It is easily seen that~$U_{i}\subset V_{i+1}$ ($U,V\in\{\P,\N\}$)  and $\P_{3}'\subset\P_{3}$.

Here is the summary of the result from~\cite{CiaGalTer08,CiaGalTer17} that we use:
In the presence of (lw), every formula in~$\P_3$ is effectively transformable to an equivalent analytic structural rule schema. In its absence, this holds for the \textit{acyclic} formulas in~$\P_{3}'$ i.e. formulas on which the above transformation terminates (see~\cite[Def. 4.11]{CiaGalTer17} for details). More formally:
\begin{definition}\label{def-amenable}
A finite set~$\mathcal{F}$ of formulas is \emph{amenable} if\\(i)~$\mathcal{F}\subseteq\P_{3}$ and left weakening $p\fus q\rightarrow p\in\mathcal{F}$, or\\(ii)~$\mathcal{F}\subseteq\P_{3}'$ consists of acyclic formulas.
\end{definition}
Call the formula $A$ amenable if $\{A\}$ is an amenable set.

\begin{theorem}[\cite{CiaGalTer08,CiaGalTer17}]\label{thm-CiaGalTer08}
\hspace{1em}
\begin{enumerate}[(i)]
\item\label{forward} Let~$\mathcal{F}$ be an amenable set. A finite set~$R_{\mathcal{F}}$ of analytic structural rule schemas is computable from~$\mathcal{F}$ such that $\HFLe+R_{\mathcal{F}}$ is a calculus for~$\FLe+\mathcal{F}$.
 
\item Every analytic structural rule extension~$\HFLe+R$ has cut-admissibility and is a calculus for some axiomatic extension of~$\FLe$ by amenable formulas.
\end{enumerate}
\end{theorem}
See Fig.~\ref{fig-str-rules} for examples of amenable formulas and the analytic structural rule schema computed from them.

The amenable extensions of~$\FLe$---and consequently the logics covered by the results in this paper---comprise a significant set of substructural logics in a formal sense: no further
axiomatic extensions can be obtained via analytic structural rule extension of~$\HFLe$~\cite{CiaGalTer17}.
Moreover, every axiomatic extension of~$\FLe$ is equivalent to an extension by $\N_3$ axioms~\cite{Jer16}
i.e. the hierarchy closes at~$\N_3$. An example of a proper $\N_3$ axiom is distributivity (it has no equivalent analytic structural rule).

\begin{figure*}
\hspace{-1.5cm}
\begin{scriptsize}
\begin{tabular}{cc}
\AxiomC{$\h |  Y_{1}, X_{1}\Ra\Pi_{1}$}
\AxiomC{$\h |  Y_{2}, X_{2}\Ra\Pi_{2}$}
\RightLabel{(com)}
\BinaryInfC{$\h |  Y_{1},X_{2}\Ra\Pi_{1} | Y_{2},X_{1}\Ra\Pi_{2}$}
\noLine
\UnaryInfC{$\landone{p\imp q}\lor \landone{q\imp p}$}
\DisplayProof
&
\AxiomC{$\h | Z_{1}, Z_{2}\Ra$}
\RightLabel{(wem)}
\UnaryInfC{$\h | Z_{1}\Ra | Z_{2}\Ra $}
\noLine
\UnaryInfC{$\landone{p\imp 0}\lor \landone{(p\imp 0)\imp 0}$}
\DisplayProof
\\[3.5em]
\AxiomC{$\h | Y_{i}, Y_{j} \Ra \Pi_{i} (0\leq i,j\leq k; i\neq j)$}
\RightLabel{($Bwk$)}
\UnaryInfC{$\h | Y_{0}\Ra \Pi_{0} | \ldots | Y_{k}\Ra \Pi_{k}$}
\noLine
\UnaryInfC{$\lor_{i=0}^{k}\landone{p_{i}\imp (\lor_{j\neq i} \, p_{j})}$}
\DisplayProof
&
\AxiomC{$\h | Y_{i}, Y_{j} \Ra \Pi_{i} (0\leq i\leq k-1; i+1\leq j\leq k)$}
\RightLabel{($Bck$)}
\UnaryInfC{$\h | Y_{0}\Ra \Pi_{0} | \ldots | Y_{k-1}\Ra \Pi_{k-1}| Y_{k}\Ra$}
\noLine
\UnaryInfC{$\landone{p_{0}}\lor \landone{p_{0}\imp p_{1}}\lor \ldots \lor \landone{(p_{0}\land\ldots \land p_{k-1})\imp p_{k}}$}
\DisplayProof
\\[3.5em]
\AxiomC{$\h | Y,X_{1}\Ra\Pi$}
\AxiomC{$\h | Y,X_{2}\Ra\Pi$}
\RightLabel{(mingle)}
\BinaryInfC{$\h | Y,X_{1},X_{2}\Ra\Pi$}
\noLine
\UnaryInfC{$\landone{p\cdot p\imp p}$}
\DisplayProof
&
\AxiomC{$\{\h | Y,X_{i_{1}},\ldots,X_{i_{m}}\Ra\Pi \textit{ s.t. } \{i_{1},\ldots,i_{m}\}\subseteq \{1,\ldots,n\} \}$}
\RightLabel{($\text{knot}^{n}_{m}$)}
\UnaryInfC{$\h | Y,X_{1},\ldots,X_{n}\Ra\Pi$}
\noLine
\UnaryInfC{$\landone{p^{n}\imp p^{m}}\text{ ($n,m\geq 0$)}$}
\DisplayProof
\end{tabular}

\caption{Analytic structural rule schema computed from the amenable formula below it. 
$\landone{A}$ denotes~$(A)\land 1$. The $\land 1$ can be dropped for~$\FLelw$ extensions.}
\label{fig-str-rules}
\end{scriptsize}
\end{figure*}

\section{Hypersequent substructural logics with weakening: decidability}\label{sec-FLew-decidability}

Throughout this section $\Omega$ is a finite set of formulas. 

A hypersequent such that every formula in it belongs to~$\Omega$ is called an \textit{$\Omega$-hypersequent}.

Let us call (lw), (EC) and (EW) the \emph{weak structural rules}.

Define the relation on $\Omega$-hypersequents $g \qo h$ iff~$h$ is derivable from~$g$ using weak structural rules.
Observe that in the witnessing derivation, (lw) will only introduce formulas from~$\Omega$, and every sequent introduced by (EW) must consist of formulas from~$\Omega$.
E.g. if $\Gamma\cup\{A\}\cup\Pi\subseteq \Omega$ then $(\Gamma,A\Ra\Pi\mid \Gamma\Ra\Pi) \qo (\Gamma,A\Ra\Pi)$ since
\begin{center}
\AxiomC{$\Gamma,A\Ra\Pi\mid \Gamma\Ra\Pi$}
\RightLabel{(lw)}
\UnaryInfC{$\Gamma,A\Ra\Pi\mid \Gamma,A\Ra\Pi$}
\RightLabel{(EC)}
\UnaryInfC{$\Gamma,A\Ra\Pi$}
\DisplayProof
\end{center}
Evidently~$\qo$  is a quasi-ordering (reflexive and transitive binary relation) but it is not a partial order because it is not antisymmetric. Indeed, if $p\in\Omega$ then
\[
(p \Ra  p)	\qo	(p \Ra  p | p \Ra  p)	 \qo	(p \Ra  p)
\]
\begin{lemma}
\label{lem:computabilityqo}
The relation $\qo$ is decidable.
\end{lemma}
\begin{proof}
It suffices to observe that for hypersequents $g,h$: $g\qo h$ iff for every component $\Gamma\Ra \Pi$ in $g$, there is a component $\Gamma'\Ra\Pi$ in $h$ such that $\Gamma$ is a submultiset of $\Gamma'$ (i.e. for each formula, its multiplicity in~$\Gamma'$ $\geq$ its multiplicity in~$\Gamma$).
The left-to-right direction is by induction on the number of weak structural rules witnessing $g\qo h$. Right-to-left: transform every $\Gamma\Ra\Pi$ in~$g$ into the corresponding $\Gamma'\Ra\Pi$ in~$h$ by (lw). Now obtain~$h$ by using (EW) to insert missing components and (EC) to remove unwanted multiplicities.
\end{proof}

\subsection{Defining $(S_{i})$ and establishing its completeness}

Throughout this subsection we take $S$ to be a finite set of $\Omega$-hypersequents. 
Also let~$\mathcal{\C}$ denote the extension $\HFLelw+R$ by any finite set~$R$ of analytic structural rule schemas. 

Define $\maxelesize{S}=\max\{\hypsize{h}|h\in S\}$ i.e. the maximum of the number of symbols in a hypersequent in~$S$. Also define~$\maxelesize{\C}=\max\{\hypsize{r}|r\in\C\}$ i.e. the maximum of the number of symbols in a rule schema in~$\C$.


A hypersequent is $S$-\textit{\thin} if 
\begin{itemize}
\item no component in it has an antecedent of cardinality $> \maxelesize{S}.\maxelesize{\C}.|\Omega|$, and 
\item no sequent has multiplicity $>\maxelesize{\C}$.
\end{itemize}
In the first condition the `$\mply$' stands for multiplication.

\begin{definition}\label{def-WI}
$\WI(S,\Omega,\C)$ $=$ $\{ h|$ $h$ is an $\Omega$-hypersequent, and $h$ is $S$-{\thin}, and $h$~is the conclusion of a rule instance (of some rule schema from $\C$) with premises $h_1,\ldots ,h_n$ such that there exists $\{\minus{h_1},\ldots ,\minus{h_n}\}\subseteq S$ and $\minus{h_i} \qo h_i$ for each~$i$ $\}$.

Define $\WI^{\infty}(S,\Omega,\C)$ by deleting the condition ``$h$ is $S$-{\thin}" from the above definition.
\end{definition}
Def.~\ref{def-WI} can be depicted as follows (every hypersequent below is taken to be an $\Omega$-hypersequent):
\begin{center}
\AxiomC{$\minus{h_{1}}\in S$}
\noLine
\UnaryInfC{$\vdots$}
\noLine
\UnaryInfC{$h_{1}$}
\AxiomC{$\cdots$}
\noLine
\UnaryInfC{weak structural rules}
\noLine
\UnaryInfC{restricted to~$\Omega$ i.e. $\qo$}
\noLine
\UnaryInfC{$\cdots$}
\AxiomC{$\minus{h_{n}}\in S$}
\noLine
\UnaryInfC{$\vdots$}
\noLine
\UnaryInfC{$h_{n}$}
\TrinaryInfC{$h$}
\noLine
\UnaryInfC{$h\in \WI^{\infty}(S,\Omega,\C)$; $h\in \WI(S,\Omega,\C)$ iff~$h$ is $S$-{\thin}}
\DisplayProof
\end{center}
$\WI^{\infty}(S,\Omega,\C)$ is infinite for non-empty~$S$ and~$\Omega$ because e.g. (EW) could be used to instantiate the hypersequent-variable in the premises with more and more components (hence yielding larger and larger conclusions). In contrast, we shall see that the $S$-{\thin} condition ensures computability of $\WI(S,\Omega,\C)$ from~$S$.

Let~$h(s)$ denote the multiplicity of sequent~$s$ in the hypersequent~$h$.
Its \textit{$2$-reduct} $\reduct{h}$ is the hypersequent such that the multiplicity of~$s$ is $\minimum(h(s),2)$.
A derivation is \textit{$2$-reduced} if every hypersequent in it~$g$ is $2$-reduced i.e. $\reduct{g}=g$.
\begin{lemma}\label{lem-tworeduced}
If $g\qo h$ then there is a $2$-reduced derivation witnessing $\reduct{g}\qo \reduct{h}$.
\end{lemma}
\begin{proof}
Let~$d$ be the derivation (consisting of weak structural rules) that witnesses $g\qo h$.
Induction on the height of~$d$. If the height is~$1$ then $g=h$ so $\reduct{g}=\reduct{h}$ and the claim is immediate.
Suppose that the height is~$k+1$. Consider the last rule in~$d$. 

Suppose it is (EC) taking $s|s|h'$ to~$s|h'$. If $h'(s)=0$ then by induction hypothesis IH $\reduct{g}\qo s|s|\reduct{h'}$ has a $2$-reduced derivation. Now apply (EC) to get $s|\reduct{h'}$ (this is $\reduct{s|h'}$ as required). If $h'(s)=1$ then by IH $\reduct{g}\qo s|\reduct{h'}$ (this is $\reduct{s|h'}$) has a $2$-reduced derivation. If $h'(s)\geq 2$ then by IH $\reduct{g}\qo \reduct{h'}$ (this is $\reduct{s|h'}$) has a $2$-reduced derivation.

Suppose it is (EW) taking $h'$ to~$s|h'$. If $h'(s)\leq 1$ then by induction hypothesis $\reduct{g}\qo \reduct{h'}$ has a $2$-reduced derivation. Now apply (EW) to get $s|\reduct{h'}$ (this is $\reduct{s|h'}$). If $h'(s)\geq 2$ then $\reduct{g}\qo \reduct{h'}$ (this is $\reduct{s|h'}$) has a $2$-reduced derivation.

Suppose that the last rule is~(lw) taking $h'|X\Ra\Pi$ to $h'|X,A\Ra\Pi$. By the induction hypothesis $\reduct{g}\qo \reduct{h'|X\Ra\Pi}$ has a $2$-reduced derivation.
If $h'(X,A\Ra\Pi)\geq 2$ then apply (EC) to $\reduct{h'|X\Ra\Pi}$ to make the $X,A\Ra\Pi$-multiplicity~$1$ and then apply (lw) taking $X\Ra\Pi$ to $X,A\Ra\Pi$. If $h'(X,A\Ra\Pi)\leq 1$ then apply (lw) to $\reduct{h'|X\Ra\Pi}$ taking $X\Ra\Pi$ to $X,A\Ra\Pi$.
In each of the two above cases: if $h'(X\Ra\Pi)\leq 1$ then the hypersequent obtained is already $\reduct{h}$; else if $h'(X\Ra\Pi)\geq 2$ then apply (EC) to take the $X\Ra\Pi$-multiplicity in the hypersequent obtained to~$2$, and so obtain a $2$-reduced derivation of~$\reduct{h}$.
\end{proof}

\begin{lemma}
\label{lem:constr}
The function $S\mapsto \WI(S,\Omega,\C)$ is computable.
\end{lemma}
\begin{proof}
Call a hypersequent \textit{\slim} if no component in it has an antecedent of cardinality $> \maxelesize{S}.\maxelesize{\C}^2.|\Omega|$. Evidently each $S$-{\thin} hypersequent is also {\slim}. 

Recall that each element of~$\WI(S,\Omega,\C)$ is a {\thin} conclusion of a rule instance from~$\C$.
We first show that the premises of such a rule instance must be {\slim}. Suppose towards a contradiction that the antecedent of some component in a premise of the rule instance has cardinality larger than $\maxelesize{S}.\maxelesize{\C}^{2}.|\Omega|$. Either this component is in the instantiation of the hypersequent-variable, or
the instantiation of some multiset-variable in that premise has cardinality larger than $\maxelesize{S}.\maxelesize{\C}.|\Omega|$ (using the fact that the number of schematic-variables in a rule schema is bounded by~$\maxelesize{\C}$).
Because this hypersequent-/multiset-variable occurs also in the conclusion of the rule (Def.~\ref{def-linear-subvariable}, subvariable property), the conclusion of the rule instance would not be {\thin}, contradicting the definition of $\WI(S,\Omega,\C)$.

Let~$N$ denote the number of different {\slim} \emph{sequents} that can be built from~$\Omega$. 
If a premise instance has more than~$\maxelesize{\C}+N.\maxelesize{\C}$ components---noting that the number of components in each rule schema is bounded by~$\maxelesize{\C}$---then the instantiation of the hypersequent-variable would contain more than~$N.\maxelesize{\C}$ components, and hence more than $\maxelesize{\C}$ copies of the same sequent, so the conclusion of the rule instance would not be {\thin}.

The above two paragraphs show that if a hypersequent is in~$\WI(S,\Omega,\C)$ then it is the $S$-{\thin} conclusion of a rule instance from~$\C$ with premises that are {\slim} and do not contain a sequent with multiplicity $>\maxelesize{\C}+N.\maxelesize{\C}$. 

For each $h\in S$ that is {\slim}, define the tree~$\tau(\reduct{h})$ of $2$-reduced hypersequents whose root is~$\reduct{h}$ and the set of children of an interior node~$g$ is the set of hypersequents that are {\slim} and $2$-reduced and can be obtained by applying a single weak structural rule (restricted to~$\Omega$) to~$g$, omitting hypersequents that already appear on the path from the root to~$g$.
Due to this `omitting condition' the height of~$\tau(\reduct{h})$ is at most~$3^{N}$---each of the $N$ {\slim} sequents has multiplicity in $\{0,1,2\}$---and hence $\tau(\reduct{h})$~is finite and computable.

Assume now that $h'$ is a \textit{potential premise} for an element of $\WI(S,\Omega,\C)$, i.e. $h'$ is {\slim}, does not contain a sequent with multiplicity $>\maxelesize{\C}+N.\maxelesize{\C}$, and there exists $h\in S$ such that $h\qo h'$. We have the following derivation composed of weak structural inferences:
\[
h\overset{\text{(EC) rules}}\qo \reduct{h}\overset{\text{$2$-reduced derivation Lem.~\ref{lem-tworeduced}}}{\qo} \reduct{h'}\overset{\text{(EW) rules}}{\qo} h'
\]
Since~$h'$ is {\slim}: every hypersequent in this derivation must be {\slim} (including~$h$!), since weak structural rules cannot decrease the cardinality of a component's antecedent. Thus $\tau(\reduct{h})$ is defined and~$\reduct{h'}$ is a node in it.
We can generate all such~$h'$ (in fact there are at most~$N^{\maxelesize{\C}+N.\maxelesize{\C}+1}$ possibilities). 

$\WI(S,\Omega,\C)$ is the set of $S$-{\thin} conclusions of rule instances from~$\C$ whose premises are among these potential premises. Evidently this set is computable.
\end{proof}

\begin{lemma}
\label{lem:WIinfty}
If $h \in \WI^{\infty}(S,\Omega,\C)$ then there exists some $h' \in \WI(S,\Omega,\C)$ such that $h' \qo  h$.
\end{lemma}
\begin{proof}
By definition $h$ is the conclusion of a rule instance~$I(r)$ of some rule schema $r \in \C$ with premises $h_1,\ldots,h_n$ such that 
$\{\minus{h_1},\ldots,\minus{h_n}\}\subseteq S$ and for each~$i$: $\minus{h_i} \qo  h_i$.
If~$h\in \WI(S,\Omega,\C)$ there is nothing to do. So suppose that $h\in \WI^{\infty}(S,\Omega,\C)\setminus \WI(S,\Omega,\C)$. 
Therefore~$h$ is not $S$-{\thin}.

We will construct a new rule instance $I'(r)$~whose premises are obtained from~$S$ under~$\qo$ and with conclusion~$h'\in \WI^{\infty}(S,\Omega,\C)$ satisfying $h'\qo  h$ and $\hypsize{h'}<\hypsize{h}$. The result then follows by iterating this construction until a rule instance~$I\phantom{}'\phantom{}^{\cdots}\phantom{}'(r)$ is reached whose conclusion is in $\WI(S,\Omega,\C)$. In particular, termination is guaranteed because the number of formulas in the conclusion of the new rule instance strictly decreases with each iteration.

Let us construct this~$I'$. Since~$h$ is not $S$-{\thin}, either:

$\blacktriangleright$ Some component~$s$ in~$h$ has an antecedent of cardinality $> \maxelesize{S}.\maxelesize{\C}.|\Omega|$: 

Some formula $A\in\Omega$ must then occur in the antecedent of~$s$ with multiplicity $> \maxelesize{S}.\maxelesize{\C}$. 
Either~$s$ is in the instantiation of the hypersequent-variable (case~i), or else $s$~is an active component (i.e. the component not in the instantiation of the hypersequent-variable) of the rule instance (case~ii). In the latter case---since the number of schematic-variables in any rule schema is $\leq \maxelesize{\C}$, and using the linear conclusion (Def.~\ref{def-linear-subvariable}) of~$\C$---there is some multiset-variable~$M$ such that $I(M)$ contains $\geq \maxelesize{S}+1$ occurrences of~$A$.

In case~i, there is a corresponding component of~$s$ in each premise~$h_{k}$; call it a \textit{marked-component}. In case~ii, call each component in each premise~$h_{k}$ that corresponds to a component containing~$M$ in the rule schema a \textit{marked component}. Let~$d_{k}$ be the derivation comprising of weak structural rules witnessing $\minus{h_{k}}\qo h_{k}$. Extend the definition of~\textit{marked-components} to the smallest set of components in the derivation~$d_{k}$ as follows: if the active component in the conclusion of (lw)  or (EC) is a marked-component, then so are the active component(s) in the premise; for every non-active component in the conclusion that is a marked-component, the corresponding component in each premise is a marked-component.
\begin{flushleft}
\textbf{Claim.} For every hypersequent~$g$ in~$d_{k}$: there is~$g^{A}$ identical to~$g$ except that the number of occurrences of the formula~$A$ in the antecedent of each marked-component is exactly~$\maxelesize{S}$, and $\minus{h_{k}}\qo g^{A}$.
\end{flushleft}
Induction on the number of rules in~$d_{k}$. If the number of rules is~$0$ (i.e. $h_{k}=\minus{h_{k}}$) then~$h_{k}\in S$ so every component in~$h_{k}$ has size~$\leq\maxelesize{S}$. Apply (lw) with~$A$ to~$\minus{h_{k}}$ as much as required in order to obtain~$h_{k}^{A}$ (hence $\minus{h_{k}}\qo h_{k}^{A}$). 

Inductive case. If the last rule is~(EC) then apply the induction hypothesis to the premise and then apply~(EC).
If the last rule is~(EW) and the introduced component is not a marked-component then apply the induction hypothesis to the premise and apply~(EW) unchanged; if it is a marked-component then use (EW) to introduce a variant where the number of occurrences of~$A$ is exactly~$\maxelesize{S}$. If the last rule is~(lw) then apply the induction hypothesis to the premises and reapply~(lw) only if it does not introduce an occurrence of~$s$ into a $M$-component. This establishes the claim.

If we had case~i then define the instantiation~$I'$ as~$I$ except the marked-component now contains exactly $\maxelesize{S}$ copies of~$A$ in the antecedent, rather than $> \maxelesize{S}.\maxelesize{\C}$ as in $I(r)$. So $h_k'=h_k^A$ for each premise~$h_{k}'$ of~$I'(r)$.

If we had case~ii then define the instantiation~$I'$ as~$I$ except~$I'(M)$ now contains exactly $\maxelesize{S}$ copies of~$A$, rather than $\geq \maxelesize{S}+1$ as in $I(r)$. Each premise~$h_{k}'$ of~$I'(r)$ is the same as~$h_{k}^{A}$ except that the marked-components have $\geq \maxelesize{S}$ copies of~$A$ (the exact number depends on the multiplicity of $M$ in the marked component) rather than exactly~$\maxelesize{S}$ copies of~$A$.

In both cases, it follows that~$h_{k}'$ is identical to $h_k^A$ or can be obtained from~$h_{k}^{A}$ by (lw), and so $h_k^-\qo h_k^A \qo h_k'$. The conclusion~$h'$ of~$I'(r)$ is the same as~$h$ but with fewer occurrences of~$A$. So $h'\in \WI^{\infty}(S,\Omega,\C)$ and $\hypsize{h'}<\hypsize{h}$.


%

$\blacktriangleright$ Some sequent~$s$ in~$h$ has multiplicity $>\maxelesize{\C}$:

Since the number of components in the conclusion of every rule schema is $\leq \maxelesize{\C}$, the instantiation of the hypersequent-variable in the conclusion of $I(r)$ must have the form $s|s|g$. This hypersequent-variable occurs in every premise of the rule schema. Apply~(EC) to each premise~$h_{k}$ to convert $s|s|g$ to~$s|g$ and call this~$h_{k}^*$. 

The rule instance~$I'(r)$ from premises $h_1^*,\ldots,h_n^*$ has conclusion~$h'$ that is the same as~$h$ but with one less component of~$s$. So $h'\in \WI^{\infty}(S,\Omega,\C)$ and $\hypsize{h'}<\hypsize{h}$.
\end{proof}

\begin{definition}[derivability sets $S_i$ of $\HFLelw+R$ wrt $\Omega$]\label{def-derive-sets}
Define $S_{0}$ as the set of instances of initial sequent schemas in~$\C=\HFLelw+R$ for
\begin{itemize}
\item formula-variables instantiated using elements from~$\Omega$, 
\item succedent-variables instantiated by an element in~$\Omega$ or empty, and 
\item hypersequent- and multiset-variables instantiated as empty.
\end{itemize}
\begin{multline*}
\text{For $i>0$ define }S_{i+1}:=S_i \cup \{ h \in \WI(S_i,\Omega,\C) |\\ \text{there does not exist $h' \in S_i$ s.t. $h' \qo  h$} \}
\end{multline*}
\end{definition}
Since~$\Omega$ and~$\C$ are finite, $S_{0}$ is a finite set consisting of elements like $A\Ra A$, $\Ra 1$, $0\Ra$, $\bot\Ra$, and $\bot\Ra A$ for each $A\in\Omega$. Also: $S_i \subseteq S_{i+1}$ for every~$i$.

\begin{lemma}
\label{lem:SN}
Let~$h$ be a hypersequent and let $\Omega$ be any finite set containing all subformulas of $h$.
If~$h$ is derivable in $\C=\HFLelw  + R$ then there is~$N$ and~$h' \in S_N$ such that~$h' \qo  h$.
\end{lemma}
\begin{proof}
Induction on the height of the derivation of~$h$.

If $h$ has a derivation of height $1$ then it is an instance of an initial rule schema. Therefore it is obtainable from~$S_0$ by (lw) and (EW) as required. 

Inductive case. Suppose $h$ has a derivation of height $>1$ with last rule~$r$. By the subformula property, every subformula in each premise~$h_{k}$ is in~$\Omega$. By the induction hypothesis applied to the $k^{\text{th}}$ premise, there exists $N_{k}$ and a hypersequent $h_k' \in S_{N_{k}}$ s.t. $h_k' \qo  h_k$. 
Since $S_i \subseteq S_{i+1}$ for every~$i$, every $h_{k}'\in S_N$ for $N:=\max\{ N_{k}\}$.
Therefore~$h\in \WI^{\infty}(S_N,\Omega,\C)$ by Def.~\ref{def-WI}. By Lem.~\ref{lem:WIinfty}, there exists $\minus{h}\in \WI(S_N,\Omega,\C)$ such that $\minus{h}\qo h$. If $h^-$ is reachable from $S_N$ via $\qo$ then so is $h$ by transitivity and the claim follows. Otherwise by definition $h^-\in S_{N+1}$, and the claim follows.
\end{proof}

\subsection{Stability of $(S_{i})$ and decidability}

\renewcommand{\qo}{\rightbarharpoon_{\Omega}} 

The finite powerset of $k$-tuples is defined
\[
\mathcal{P}_f(\mathbb{N}^{k})=\{U\in \mathcal{P}(\mathbb{N}^{k}) | \text{$U$ is finite} \}
\]
The \textit{majoring ordering} is used in the proof of Lem.~\ref{lem-stabilises}. 
\begin{definition}[majoring ordering] 
Let $X,Y\in P_f(\mathbb{N}^{k})$. Let~$\leq$ denote the usual componentwise ordering on $k$-tuples of natural numbers.
The \emph{majoring ordering} is defined
\[
\text{$X \maj Y$ iff $\forall x\in X\exists y\in Y(x\leq y)$}
\]
\end{definition}
We say that $(A,\leq_{A})$ is a \textit{well-quasi-ordering} (wqo) if~$\leq_{A}$ is a quasi-ordering on~$A$ and for every infinite sequence $(a_i)$ over $A$ there exists $i,j$ ($i<j$) s.t. $a_i \leq_{A} a_j$.
For $d>0$ and
\[
\text{$(X_{1},\ldots,X_{d}), (Y_{1},\ldots,Y_{d})\in (P_f(\mathbb{N}^{k}))^{d}$ (written $P_f(\mathbb{N}^{k})^{d}$)}
\] 
define the \textit{$d$-majoring ordering}
\begin{multline*}
\text{$(X_{1},\ldots,X_{d})\maj^{d} (Y_{1},\ldots,Y_{d})$ iff}\\ \text{$X_{i}\maj Y_{i}$ for every~$i$ ($1\leq i\leq d$)}
\end{multline*}
\begin{theorem}\label{thm-majoring-wqo}
Let $k,d>0$. Then\\
(i)~$(P_f(\mathbb{N}^{k}), \maj)$ is a well-quasi-ordering.\\
(ii)~$(P_f(\mathbb{N}^{k})^{d}, \maj^{d})$ is a well-quasi-ordering.
\end{theorem}
\begin{proof}
(i) See Section 2 of~\cite{Linearise}. (ii) Given two wqos $(A,\leq_{A})$ and $(B,\leq_{B})$ define the following relation $\leq_{A \times B}$ on elements of $A \times B$: $(a,b) \leq_{A \times B} (a',b')$ iff $a \leq_A a'$ and $b \leq_B b'$. Section 2 of~\cite{ICALP} establishes that
$(A \times B, \leq_{A \times B})$ is also a wqo. Since $(P_f(\mathbb{N}^k),\maj)$ is a wqo, it follows that $(P_f(\mathbb{N}^{k})^{d}, \maj^{d})$ is also a wqo.
\end{proof}
\subsubsection*{From $\Omega$-hypersequent to an element of $\powf{\nn^{|\Omega|}}^{|\Omega|+1}$}
Let~$\Omega$ be a finite set of formulas and let $h$ be an $\Omega$-hypersequent. Fix any enumeration $F_1,\dots,F_{|\Omega|}$ of $\Omega$ and let $F_0$ denote the empty formula. 
Let $h_i$ ($0\leq i\leq |\Omega|$) denote the subhypersequent $X_i^1 \Ra F_i \ | \ X_i^2 \Ra F_i \ | \ \dots \ | \ X_i^n \Ra F_i$ consisting of exactly those components in~$h$ whose succedent is~$F_i$. Set $h_i^\# = \{x_i^1,\dots,x_i^n\} \in \pow{|\Omega|}$ where the value of the $k^{\text{th}}$ coordinate of the $|\Omega|$-tuple $x_i^l$ is taken to be the multiplicity of $F_k$ in the multiset~$X_i^l$.

Then $h^\# = (h_0^\#,h_1^\#,\dots,h_{|\Omega|}^\#) \in \powf{\nn^{|\Omega|}}^{|\Omega|+1}$. 
\begin{example}\label{eg-hash}
Let~$h$ be the $\{p,q,p\land q\}$-hypersequent
\[
\Ra p| p\land q, p, p \Ra p | q\Ra | q\Ra
\]
Let us use the enumeration $1\mapsto p$, $2\mapsto q$, $3\mapsto p\land q$. Then
\begin{align*}
&h_0	=q\Ra|q\Ra	&& h_1=\phantom{q}\Ra p| p\land q,p,p\Ra p\\
&h_0^\#=\{(0,1,0)\}	&& h_1^\#=\{(0,0,0),(2,0,1)\}
\end{align*}
$h_{2}$ and~$h_{3}$ are empty hypersequents so $h_2^\#=h_3^\#=\emptyset$. So
\[
h^\# = \big(\, \{(0,1,0)\} \,,\, \{(0,0,0),(2,0,1)\} \,,\, \emptyset \,,\, \emptyset \,\big) \in \powf{\nn^3}^{4}
\]
\end{example}

\begin{lemma}
\label{lem:corr}
Let $h,g$ be $\Omega$-hypersequents. Then $h\qo g$ iff $h^\# \maj^{|\Omega|+1} g^\#$.
\end{lemma}
\begin{proof}
We have that $h\qo g$ iff for every component $\Gamma\Ra\Pi$ in $h$ there is a component $\Gamma'\Ra\Pi$ in $g$ with $\Gamma\subseteq \Gamma'$ (cf.\ Lemma~\ref{lem:computabilityqo}). Equivalently, for every $0\le i\le |\Omega|$ and every component $\Gamma\Ra F_i$ in $h$ there is a component $\Gamma'\Ra F_i$ in $g$ with $\Gamma\subseteq \Gamma'$. Reformulated in terms of the $\#$ function this means that $h_i^\#\maj g_i^\#$ for every $0\le i\le |\Omega|$, or equivalently $h^\# \maj^{|\Omega|+1} g^\#$.
\end{proof}

\begin{lemma}[stability]\label{lem-stabilises}
Let $(S_i)$ be the sequence of derivability sets from Def.~\ref{def-derive-sets}. There exists~$N$ s.t. $S_{N+1}(\Omega)=S_N(\Omega)$.
\end{lemma}
\begin{proof}
Suppose not. Then $S_i \subset S_{i+1}$ for every~$i$.
Hence for any~$h_{0}\in S_{0}$ there is a sequence~$(h_{i})$ such that $h_{i+1} \in S_{i+1} \setminus S_i$ for every~$i$. Consider arbitrary $i,j$ with $i<j$ and suppose that $h_{i}\qo h_{j}$. 
Since $h_{i}\in S_{i}$ and $S_{i}\subseteq S_{i+1}$ we have $h_{i}\in S_{j-1}$ and hence $h_{j}\not\in S_{j}$ by the ``there does not exist\ldots" condition in Def.~\ref{def-derive-sets}. This is a contradiction so we conclude that $h_{i}\not\qo  h_j$.

By Lem.~\ref{lem:corr}, $h \not \qo g$ implies $h^\# \not \maj^{|\Omega|+1} g^\#$.
So $(h_i^\#)$ is a sequence in $\powf{\nn^{|\Omega|}}^{|\Omega|+1}$ such that for every~$i,j\in\mathbb{N}$ with $i<j$: $h_i^\# \not\maj^{|\Omega|+1}  h_j^\#$. This contradicts that~$\maj^{|\Omega|+1}$ is a wqo.
\end{proof}
We are ready to prove the main result of this section.
\begin{theorem}\label{thm:main-theorem}
Every analytic structural rule extension~$\C$ of $\HFLelw$ is decidable.
\end{theorem}
\begin{proof}
Let $h$ be a hypersequent and $\Omega$ the set of its subformulas. Then $h$ is derivable iff there exists some $N$ and $h' \in S_N$ such that $h'  \qo h$ (right to left is trivial, the other direction is Lem.~\ref{lem:SN}). 
Evidently we can compute~$S_{0}$. 
Therefore by Lem.~\ref{lem:constr} we can compute a finite initial segment of the sequence $(S_i)$ until $S_{N+1}=S_{N}$ (Lem.~\ref{lem-stabilises}). Finally, decide (Lem.~\ref{lem:computabilityqo}) if there is some $h' \in S_{N}$ such that $h' \qo h$.
\end{proof}

\section{Hypersequent logics with contraction: decidability}\label{sec-FLec-decidability}

The following result was established by Ramanayake.
\begin{theorem}[\cite{Ram20LICS}]\label{thm-FLec-extensions}
Every analytic structural rule extension of $\HFLec$ is decidable.
\end{theorem}
An argument by contradiction---summarised below---using the infinite Ramsey theorem (IRT) was used in that work to establish the finiteness of the backward proof search tree rooted at the input hypersequent.

The quasi-order~$\qohyp$ on hypersequents is defined\footnote{\cite{Ram20LICS} uses the equivalent partial ordering of this quasi-ordering. It is obtained by considering the quotient classes. The quasi-ordering is used here for uniformity with the approach in the previous section.} as $g\qohyp h$ iff $g$~can be obtained from~$h$ by repeated applications of~(c), (EC), and (EW).
A proof search tree has the \textit{irredundancy property} if whenever~$g$ and~$h$ appear on the same branch with~$g$ closer to the root, then~$g\not\qohyp h$. 
Completeness of the (finitely branching) irredundant proof search tree is established in~\cite[Theorem 5.1]{Ram20LICS}.
Now suppose that this tree is not finite. Then by K\"{o}nig's lemma it must contain an infinite branch. Hence there is an infinite sequence~$(h_{i})$ of hypersequents with $h_{i}\not\qohyp h_{j}$ for every $i<j$. The IRT was then used to obtain a sequence~$(g_{i})$ of hypersequents whose \textit{every component has the same succedent}, and $g_{i}\not\qohyp g_{j}$ for every $i<j$. This implies an infinite sequence~$(g_{i}^{\#})$ in~$\mathcal{P}_f(\mathbb{N}^{k})$ 
($k$ is the cardinality of the set of subformulas of formulas in the input hypersequent) such that $g_{i}\not\min g_{j}$ for every $i<j$, violating the well-ordering property of the minoring ordering~$\min$ defined below. This yields the desired contradiction.
\begin{definition}[minoring ordering] 
Let $X,Y\in P_f(\mathbb{N}^{k})$. Let~$\leq$ denote the usual componentwise ordering on $k$-tuples of natural numbers
and define the \emph{minoring ordering}
\[
\text{$X \min Y$ iff $\forall y\in Y\exists x\in X(x\leq y)$}
\]
\end{definition}
We are not aware of complexity bounds associated with the IRT.
To obtain complexity upper bounds (following section), we therefore refine the argument in~\cite{Ram20LICS} by using the following \textit{$d$-minoring ordering}~$\min^{d}$ in place of the IRT.
For $d>0$ and $(X_{1},\ldots,X_{d}), (Y_{1},\ldots,Y_{d})\in P_f(\mathbb{N}^{k})^{d}$, define
\begin{multline*}
\text{$(X_{1},\ldots,X_{d})\min^{d} (Y_{1},\ldots,Y_{d})$ iff}\\ \text{$X_{i}\min Y_{i}$ for every~$i$ ($1\leq i\leq d$)}
\end{multline*}
\begin{theorem}\label{thm-minoring-wqo}
Let $k,d>0$. Then\\
(i)~$(P_f(\mathbb{N}^{k}), \min)$ is a well-quasi-ordering.\\
(ii)~$(P_f(\mathbb{N}^{k})^{d}, \min^{d})$ is a well-quasi-ordering.
\end{theorem}
\begin{proof}
(i) See~\cite{Ram20LICS}. (ii) Similar proof as Thm.~\ref{thm-majoring-wqo}(ii). 
\end{proof}
\begin{proof}[Argument for Thm.~\ref{thm-FLec-extensions} without IRT]
Let~$\HFLec  + R$ be an arbitrary analytic structural rule extension of~$\HFLec$.
As argued in the proof of~\cite[Theorem 5.1]{Ram20LICS}: $h$ is derivable in~$\HFLec  + R$
iff $h$ has an irredundant derivation in~$\HFLec  + R$. 
Construct a proof search tree of~$h$ as follows:
Place~$h$ at the root. Repeatedly, for each hypersequent in
the tree, place all the premises of all possible rule instances
as its children, omitting those rule instances that would introduce
a premise that violates irredundancy of the tree. The
proof search tree is finitely branching since there are only
finitely many possible rule instances that apply to a given
hypersequent conclusion. If the length of all the branches in the tree are not bounded by some value under this construction, then there exists a sequence~$(h_{i})$ of hypersequents built from the subformulas~$F_{1},\ldots,F_{d}$ of~$h$ such that for every~$i<j$: $h_{i}\not\qohyp h_{j}$.
It follows that~$(h_i^\#)$---the $\#$ function is defined above Eg.~\ref{eg-hash}---is a sequence in $\powf{\nn^d}^{(d+1)}$ such that for every~$i,j\in\mathbb{N}$: $h_i^\# \not \min^{d+1} h_j^\#$. This contradicts that~$\min^{d+1}$ is a wqo.

Hence the proof search tree is finite (K\"{o}nig's lemma) so the construction
terminates. The proof search tree contains a derivation as a subtree iff $h$~is derivable.
\end{proof}

\section{Complexity upper bounds}\label{sec-upper-bound}

Given a well-quasi-order~$\preccurlyeq_{A}$ on the set~$A$, a sequence $a_{0},\ldots,a_{N}$ ($a_{i}\in A$) is called a \textit{bad sequence} if for every $i<j$ it is the case that $a_{i}\not\preccurlyeq_A a_{j}$.

The main ingredient for obtaining a complexity upper bound on the decision procedures in the above sections is a bound on the maximum length of `eligible' bad sequences under the majoring and minoring orderings. Why not a bound on \textit{all} bad sequences? Because a maximum length does not exist in general: Consider the 
 usual componentwise ordering~$\le$ on $\mathbb{N}^2$.
It is well known (e.g. Dickson's lemma) that~$\le$ is a well-quasi ordering. Clearly $(1,1),(0,n),(0,n-1),\ldots,(0,0)$ is a bad sequence of length~$n+2$ for any~$n$ and hence there can be no maximum length for the bad sequences under~$\le$. Evidently, the reason is the arbitrarily large ``jump" from $(1,1)$ to $(0,n)$. 

However, if the sequence is generated by some process (e.g. proof search) we might be able to identify some bound on the magnitude of the jumps, and for such bad sequences a maximum length may exist. Figueira \textit{et al.}~\cite{LICS} and Schmitz \textit{et al.}~\cite{ICALP} identify sufficient conditions for the latter to hold: bad sequences whose sequential growth in size (defined using some norm~$|\cdot|_{A}$) is controlled by a monotone function~$g$ and starting value~$n$, i.e. bad sequences $a_0,a_1,\dots,a_N$ such that $|a_{0}|_{A}\leq n$, $|a_{1}|_{A}\leq g(n)$, $|a_{2}|_{A}\leq g(g(n))$ and so on.

\begin{definition}[\cite{ICALP,LICS}]
A \emph{normed wqo} is a wqo~$\preccurlyeq_{A}$ on a set~$A$ and a norm~$|\cdot|_{A}:A\mapsto\mathbb{N}$ that is \emph{proper} in the sense that $\{a :  |a|_{A} \leq n\}$ is finite for every~$n\in\mathbb{N}$.
\end{definition}

\begin{definition}[\cite{ICALP,LICS}]
A \emph{control function} is any function $g : \nn \to \nn$ that is monotone and $g(x) \ge x$ for all $x \in \nn$. Let $g$ be any control function and let $n \in \nn$. A sequence $a_{0},a_{1},\dots$ over elements
of~$A$ is called a \emph{$(g,n)$-controlled bad sequence} over 
the $|\cdot|_{A}$ normed wqo $\preccurlyeq_{A}$ iff
\begin{itemize}
	\item There is no $i < j$ such that $a_i \preccurlyeq_A a_j$, and
	\item $|a_{i}|_{A} \le g^i(n)$, where $g^i(n)$ denotes $i$-fold composition of $g$ with itself.
\end{itemize}
\end{definition}

\begin{lemma}[\cite{ICALP,LICS}]
Let~$(A,\preccurlyeq_{A},|\cdot|_{A})$ be a normed wqo, $g$ a control function and $n \in \nn$. There is a $(g,n)$-controlled bad sequence of finite maximum length.
\end{lemma}
\begin{proof}
Consider the tree whose nodes are $(g,n)$-controlled bad sequences such that the root is the empty sequence, 
and the set of children of a node~$\underline{x}$ are those $(g,n)$-controlled bad sequences of the form $\underline{x},a$ ($a\in A$).
Since $|a|_{A}\leq g^{|\underline{x}|}(n)$ and $|\cdot|_{A}$ is a proper norm, it follows that the tree is finitely branching. The tree has no infinite branch because~$\preccurlyeq_{A}$ is a wqo. Hence by K\"{o}nig's lemma it is finite and so there is a branch (and hence a $(g,n)$-controlled bad sequence) of maximum length.
\end{proof}

For $n\in\mathbb{N}$ and control function~$g$, let $L_{A,\preccurlyeq_A,g}(n)$ (\textit{length function}) be the length of the longest $(g,n)$-controlled bad sequence over the normed wqo $(A,\preccurlyeq_{A},|\cdot|_{A})$.

\subsection*{Norm over finite powersets}

Define these norms over $\nn^k$, $\pow{k}$, and $(P_f(\mathbb{N}^k))^{d}$:
\begin{itemize}
\item $\|\cdot\|$ of $x = (x_1,\dots,x_k) \in \nn^k$ as the maximum of $x_1,\dots,x_k$
\item $\|\cdot\|$ of $X\in\pow{k}$ as $\max\left(\{|X|\}\cup \{\|x\| \,|\, x \in X\}\right)$. Here $|X|$ is
the cardinality of the set $X$
\item $\|\cdot\|$ of $(X^1,\dots,X^d)\in P_f(\mathbb{N}^k)^{d}$ as the maximum of $\|X^{1}\|,\ldots,\|X^{d}\|$
\end{itemize}
 It is easy to check that the latter norm is proper and hence that $(P_f(\mathbb{N}^k)^{d}, \min^{d}, \|\cdot\|)$ and $(P_f(\mathbb{N}^k)^{d}, \maj^{d}, \|\cdot\|)$ are normed wqos.
We are ready to talk about length functions of these wqos and their asymptotic computational complexity. For this we 
will need the \emph{fast-growing function hierarchy} and the \emph{fast-growing complexity classes}.

\subsection*{Fast-growing complexity classes}

A large collection of problems in verification, automata theory, formal languages and logic have running times that grow much faster than any \emph{elementary function} and indeed any primitive recursive function.
To compare computational problems that fall into this category, a notion of \emph{fast-growing complexity classes} based on \emph{ordinals} is used. We will not give a complete definition but only state those definitions and facts that are essential to obtain our upper bound on the running time. See~\cite{BeyondElem} for details and a survey of these classes.

First a hierarchy of fast-growing functions $\{F_\alpha\}_{\alpha < \epsilon_0}$ is defined, where for each ordinal $\alpha < \epsilon_0$, we have a function $F_{\alpha} : \nn \to \nn$. 
Using these functions, the extended Grzegorczyk hierarchy $\{\mathscr{F}_\alpha\}_{\alpha < \epsilon_0}$
is defined, where for each $\alpha < \epsilon_0$, we have a collection
of functions $\mathscr{F}_\alpha$. 
Let $\mathbf{F}^*_{\alpha}$ denote the class $\bigcup_{\beta < \alpha}  \bigcup_{p \in \mathscr{F}_\beta} \{F_{\alpha}(p(n))\}$ of functions of~$n$.
All that we require for our purposes are the following two facts.
\begin{lemma}[Lemma~4.6 of \cite{BeyondElem}]~\label{lem:closed-pr-functions}
	If $f : \mathbb{N} \to \mathbb{N}$ is a function in $\mathbf{F}^*_{\omega^{\omega}}$ 
	and $g_1, g_2$ are primitive recursive functions then
	the function $g_1 \circ f \circ g_2$ is also in $\mathbf{F}^*_{\omega^{\omega}}$.
\end{lemma}
\begin{theorem}[\cite{Bal20}]~\label{thm:length-functions}
	Let $g$ be any fixed primitive recursive function.
	For any fixed $d$ and $k$, the functions which map $n$ to
	$L_{\pow{d}^k,\maj^k,g}(n)$ and $L_{\pow{d}^k,\min^k,g}(n)$
	are upper-bounded by functions in the class
	$\mathbf{F}^*_{\omega^{d}}$. When $d$ and $k$ are not fixed,
	but are arguments along with the number $n$, then the functions
	$L_{\pow{d}^k,\maj^k,g}(n)$ and $L_{\pow{d}^k,\min^k,g}(n)$ are
	upper-bounded by functions in the class $\mathbf{F}^*_{\omega^{\omega}}$. 
\end{theorem} 
Define $\mathbf{F}_{\alpha}$ to be the set of decision problems that can be decided by a deterministic Turing machine in time $F_{\alpha}(p(n))$ where $n$ is the size of the input and $p$ is some function belonging to any of the ``lower'' classes $\bigcup_{\beta < \alpha}  \mathscr{F}_{\beta}$. I.e.
$$\mathbf{F}_{\alpha} = \bigcup_{\beta < \alpha } \ \bigcup_{p \in \mathscr{F}_{\beta}} \text{TIME} (F_{\alpha}(p(n)))$$
The distinction between deterministic and non-deterministic and between time and space bounds is irrelevant for 
$\mathbf{F}_{\alpha}$ with $\alpha>2$ because the class is closed under exponential functions.
Of primary interest to us are the two classes $\mathbf{F}_{\omega}$ and $\mathbf{F}_{\omega^{\omega}}$.

Informally speaking, $\mathbf{F}_{\omega}$ consists of those decision problems whose running time 
can be obtained by composing primitive recursive functions and a \emph{single application} of an Ackermannian function.
The decidability problem for~$\FLec$ is $\mathbf{F}_{\omega}$-complete~\cite{Urq99}.
Meanwhile~$\mathbf{F}_{\omega^\omega}$ consists of those decision problems whose running time can be obtained by composing
\emph{multiply recursive functions} and a single application of a \emph{hyper-Ackermannian function}.
Roughly speaking, multiply recursive functions and hyper-Ackermannian functions are higher-ordinal analogues of primitive recursive functions and Ackermannian functions respectively.

\subsection{Hypersequent substructural logics with weakening}

Let us relate the running time of the decision procedure of $\HFLelw+R$ to the maximum length of controlled
bad sequences over the majoring ordering. 

Let~$\C=\HFLelw + R$ be an analytic structural rule extension and~$h$ the arbitrary input hypersequent.
Also let~$\Omega$ be the set of subformulas of $h$. Set $d:= |\Omega|$ and $n:= \maxelesize{h}$. 

As described in Thm.~\ref{thm:main-theorem}, we first compute the sets $S_{0},S_{1},\ldots$ until we encounter the first index $N$ such
that $S_{N+1} = S_N$. Then we check if there exists $h' \in S_N$ such that $h'\qo h$. 
Notice that given the set $S_i$, we can compute $S_{i+1}$ in exponential time.
Also, notice that checking if there exists $h' \in S_N$ such that $h' \qo h$ takes at most exponential
time in the size of $S_N$. Hence, the running time of the algorithm is a primitive recursive
function of $\sum_{i=0}^{N} |S_i|$ and $n$.

We now show that the the size of $\sum_{i=0}^N |S_i|$ can be upper-bounded by a primitive recursive function of $N$ and $n$. 
Then we will show that $N$ can be upper-bounded by a function of $n$ for some function in $\mathbf{F}^*_{\omega^{\omega}}$.
Hence, by using Lem.~\ref{lem:closed-pr-functions}, we can then conclude that the running time
of the algorithm is upper-bounded by a function in $\mathbf{F}^*_{\omega^{\omega}}$ and 
so the problem is in $\mathbf{F}_{\omega^{\omega}}$.

Suppose that $h' \in \WI(S,\Omega,\C)$. From the definition of $S$-{\thin}
it follows that there is a primitive recursive function~$f'$ such that $\maxelesize{h'} \leq f'(\maxelesize{S} \mply n)$. Specifically, $h'$ contains at most $(\maxelesize{S}.\maxelesize{C}.(|\Omega| + 1))^{|\Omega| + 1}$ different sequents, each with multiplicity at most $\maxelesize{C}$.
Multiply these values with the maximum number of symbols in a component $kn(\maxelesize{S}.\maxelesize{C}.|(\Omega| + 1))$ (constant $k$ accounts for the structural symbols: comma, $\Ra$, $|$) to bound $\maxelesize{h'}$.
By a log-transformation and using $|\Omega|\leq n$, we obtain $\maxelesize{h'}\le f'(\maxelesize{S} \mply n)$ where $f'(x) = 2^{x^{c'}}$ (for some fixed $c' \ge 1$ which does not depend on~$n$). Define $g(x) = f'(x^2)$. 

Clearly there is a polynomial function $f(x) = x^c$ for some fixed $c > 1$ which does not depend on $n$ such that $\maxelesize{S_0} \le f(n)$. After all, every hypersequent in~$S_0$ contains at most two formulas, each with size $\leq n$. 
Assuming $\maxelesize{S_i} \le g^i(f(n))$:
\begin{align*}
\maxelesize{S_{i+1}} 	&\le f'(\maxelesize{S_i}\mply n) & \text{above, def. of $S_{i+1}$}	\\
 					&\le f'(g^i(f(n)) \mply n)		& \maxelesize{S_i} \le g^i(f(n))	\\
					&< f'(g^i(f(n)) \mply g^i(f(n)))	& n<g^i(f(n))\\
					&=g(g^i(f(n))) = g^{i+1}(f(n))	&
\end{align*}
So by induction $\maxelesize{S_i} \le g^i(f(n))$ holds for all $i$.	
It follows that $\sum_{i=0}^N |S_i| \le |\{h' \,|\, \maxelesize{h'} \le g^N(f(n)) \}|$.
The number of different symbols that may appear in a hypersequent in~$S_{i}$ is bounded by some polynomial~$p(n)$.
Thus $|\{h' \,|\, \maxelesize{h'} \le g^N(f(n)) \}|\leq (g^N(f(n)))^{p(n)}$. Since $g$ and $f$ are primitive recursive, it follows that $\sum_{i=0}^N |S_i|$ can be upper-bounded by a primitive recursive function of 
$N$ and $n$. 

Now let us proceed to bound $N$ in terms of $n$. Let $h_0$ be any element in the set $S_0$.
Since each $S_i \subset S_{i+1}$,
it follows that for all $1 \le i \le N$, we can find a $h_i \in S_i \setminus S_{i-1}$.
Recall that the $\#$ function is defined above Eg.~\ref{eg-hash}.  
Since norm-size $\leq$ number of symbols i.e. $\|h_i^{\#}\|\le\maxelesize{h_i}$, it follows from the above calculation that $\|h_i^{\#}\| \le g^i(f(n))$ for each $i$. Thus $h^{\#}_{0},h^{\#}_1,\ldots,h^{\#}_{N}$ is a $(g,f(n))$-controlled bad sequence of length $N$ over~$\maj^{d+1}$. 
Hence, the maximum value for~$N$ is the length of the longest $(g,f(n))$-controlled bad sequence.
By Thm.~\ref{thm:length-functions} $N$ is upper-bounded---for fixed $d$---by a function in the class $\mathbf{F}^*_{\omega^{d}}$. In general, both $d$ and $n$ will vary with the input; then by Thm.~\ref{thm:length-functions}, we have that~$N$ is upper-bounded by a function in the class $\mathbf{F}^*_{\omega^{\omega}}$.

All the other algorithmic operations are primitive recursive so by Lem.~\ref{lem:closed-pr-functions} the running time is upper-bounded by a function in $\mathbf{F}^*_{\omega^{\omega}}$. We have therefore established the following:

\begin{theorem}
The decision problem for every analytic structural rule extension~$\HFLelw + R$ is in $\mathbf{F}_{\omega^{\omega}}$.
\end{theorem}

\subsection{Hypersequent substructural logics with contraction}

Let us relate the running time of the decision procedure to the maximum length of controlled
bad sequences over the minoring ordering. 

Let~$\C=\HFLec + R$ be an analytic structural rule extension and~$h$ the arbitrary input hypersequent. Also let~$\Omega$ be the set of subformulas of $h$. Set $d := |\Omega|$
and $n := \maxelesize{h}$.

As described in Sec.~\ref{sec-FLec-decidability}, we construct a backward proof search tree rooted at~$h$
and then check if the tree contains a subtree that is a derivation of~$h$. Let $S_i$ denote the 
set of nodes at height $i$ from the root of the tree. It is clear that given the set $S_i$,
we can compute the set $S_{i+1}$ in exponential time. Also, notice that checking if the proof tree contains a subtree that is a derivation of~$h$ is at worst exponential in the size of tree. 
Letting $N$ be the length of the longest branch in the proof search tree, it then follows
that the running time of the algorithm is a primitive recursive function of $\sum_{i=0}^N |S_i|$ and $n$.

For any rule instance of $\C$, the size of the premises can be bounded by a fixed polynomial~$g$ (determined by~$\C$) in terms of the size of the conclusion. 
Hence, it can then be easily verified that if $h' \in S_i$ then $\maxelesize{h'} \le g^i(n)$ and hence $\maxelesize{S_i} \le g^i(n)$. 
By the same argument as in the previous subsection, we conclude that  $\sum_{i=0}^N |S_i|\le |\{h' \,|\, \maxelesize{h'} \le g^N(f(n)) \}|$ can be upper-bounded by a primitive recursive function of 
$N$ and $n$. 

Now, as argued in Sec.~\ref{sec-FLec-decidability}, every branch in the backward proof search tree
corresponds to a bad sequence on $\pow{d}^{(d+1)}$ under the $(d+1)$-minoring ordering. Further, we have seen that if $h'$ is a hypersequent at height $i$ from the root,
then $\|h'\| \le \maxelesize{h'}\leq g^i(n)$. It follows that every branch in the proof search tree is a $(g,n)$-controlled bad
sequence on $\pow{d}^{(d+1)}$ under the $(d+1)$-minoring ordering. Hence, by Thm.~\ref{thm:length-functions}, the length of the longest
branch $N$ under all inputs can be upper-bounded by a function in the class $\mathbf{F}^*_{\omega^{\omega}}$.

By Lem.~\ref{lem:closed-pr-functions} the running time is upper-bounded by a function in $\mathbf{F}^*_{\omega^{\omega}}$. Therefore:

\begin{theorem}
The decision problem for every analytic structural rule extension~$\HFLec + R$ is in $\mathbf{F}_{\omega^{\omega}}$.
\end{theorem}

\emph{Remark:} If $\C$ is simply $\FLec$ then the backward proof search tree that we construct will only contain sequents (not hypersequents).
Hence, from the backward proof search tree we will extract a controlled bad sequence over $\pow{d}^{(d+1)}$ where each coordinate of each
element in the sequence is a singleton set. It follows then that the sequence is actually a controlled bad sequence over $(\mathbb{N}^d)^{(d+1)}$ under
the usual product ordering. By~\cite{LICS} we then get a $\mathbf{F}_{\omega}$ upper bound for the problem, which matches
the analysis of $\FLec$ by~\cite{Urq99}.

\section{Conclusion}

Lower bounds for the considered logics is an intriguing problem that requires a different approach, namely the embedding of a problem with known complexity into the logic. Lower bounds for~$\MTL$ are of particular interest.
The $\mathbf{F}_{\omega^\omega}$ upper bound (contrast this with the $\mathbf{F}_{\omega}$ membership of~$\FLec$) is clearly related to the move from sequents to hypersequents. Although $\FLec$ has a sequent calculus with the subformula property, a hypersequent calculus is essential for most of the extensions that we consider. This motivates the search for an extension of~$\FLec$ that is in $\mathbf{F}_{\omega^\omega}\setminus \mathbf{F}_{\omega}$.

We have seen how proof search for hypersequent substructural logics can be terminated finitely by exploiting the weakening and contraction rules to prune the search tree. How about in the absence of these rules? We observe that the decision problem for uninorm (fuzzy) logic (its hypersequent calculus is $\HFLe + (com)$) is open. Note also recent work~\cite{GalStJ20} identifying many extensions of~$\FLe$ whose derivability/deducibility problem is undecidable.

%

\bibliographystyle{plain}
\bibliography{mybib-FLec-2021-Apr-18}

%
%

\end{document}